\newif\ifproceedings\proceedingstrue

\documentclass[compsoc, conference, letterpaper, 10pt, times]{IEEEtran}

\usepackage{framed}
\usepackage{minibox}
\usepackage{amsmath,amssymb,amsthm,amsfonts}
\usepackage{graphicx}
\usepackage{xcolor}
\usepackage[T1]{fontenc}
\usepackage[ruled,algosection,noend,noline]{algorithm2e}
\usepackage{multirow,array}
\SetMathAlphabet{\mathcal}{normal}{OMS}{lmsy}{m}{n}
\SetMathAlphabet{\mathcal}{bold}{OMS}{lmsy}{m}{n}
\SetSymbolFont{largesymbols}{normal}{OMX}{lmex}{m}{n}
\SetSymbolFont{largesymbols}{bold}{OMX}{lmex}{m}{n}
\DeclareMathAlphabet{\mathsc}{OT1}{cmr}{m}{sc}

\usepackage{enumitem}
\usepackage[all]{xy}
\usepackage{booktabs}
\usepackage{multirow}
\usepackage{url}
\usepackage{xspace}
\usepackage[group-separator={,}]{siunitx}
\usepackage{epigraph}

\newtheorem{theorem}{Theorem}[section]
\newtheorem{lemma}[theorem]{Lemma}
\newtheorem{definition}[theorem]{Definition}

\newenvironment{denseitemize}{
\begin{itemize}
    \setlength{\itemsep}{1pt}
    \setlength{\parskip}{0pt}
    \setlength{\parsep}{0pt}
}{\end{itemize}}

\newcommand\bi{\begin{denseitemize}}
\newcommand\ei{\end{denseitemize}}
\newcommand\ben{\begin{enumerate}}
\newcommand\een{\end{enumerate}}
\newcommand{\trm}[1]{\textrm{#1}}
\newcommand{\tbf}[1]{\textbf{#1}}

\newcommand*\dash{\ifvmode\quitvmode\else\unskip\kern.16667em\fi---%
\hskip.16667em\relax}

\newcommand{\pr}{\textrm{Pr}}

\newcommand{\A}{\mathcal{A}}
\newcommand{\N}{\mathbb{N}}
\newcommand{\F}{\mathbb{F}}
\newcommand{\secp}{\lambda}
\newcommand{\usecp}{1^{\secp}}

\newcommand{\stateinfo}{\mathsf{state}}

\newcommand{\randpick}{\mathrel{\xleftarrow{\$}}}

\newcommand{\inputs}{\mathsf{inputs}}
\newcommand{\outputs}{\mathsf{outputs}}

\newcommand{\pk}{\ensuremath{pk}}
\newcommand{\sk}{\ensuremath{sk}}

\newcommand{\round}{\mathsf{rnd}}
\newcommand{\maxr}{\mathsf{max}}

\newcommand{\tx}{\mathsf{tx}}
\newcommand{\dest}{\mathsf{to}}
\newcommand{\data}{\mathsf{data}}
\newcommand{\amount}{\mathsf{amt}}
\newcommand{\formtx}{\mathsf{FormTx}}
\newcommand{\verifytx}{\mathsf{VerifyTx}}

\newcommand{\setup}{\mathsf{Setup}}
\newcommand{\verify}{\mathsf{Verify}}
\newcommand{\prove}{\mathsf{Reveal}}
\newcommand{\coord}{\mathtt{Update}}

\newcommand{\privatestate}{\stateinfo_{\mathsf{priv}}}

\newcommand{\pubstate}{\stateinfo_{\mathsf{pub}}}

\newcommand{\target}{\mathsf{target}}
\newcommand{\txset}{\mathsf{txset}}

\newcommand{\nonce}{\mathsf{nonce}}

\newcommand{\participants}{\mathcal{P}}

\newcommand{\afrac}{f_a}

\newcommand{\globalrandao}{R}
\newcommand{\randao}{h}
\newcommand{\seed}{\ensuremath{s}}
\newcommand{\deposit}{\mathsf{Commit}}
\newcommand{\reveal}{\mathsf{Reveal}}
\newcommand{\commit}{\mathsf{Commit}}
\newcommand{\barrier}{\mathsf{barrier}}

\newcommand{\participantsnumber}{n}

\newcommand{\hashmax}{H_{\mathsf{max}}}
\newcommand{\txdep}{\tx_{\mathsf{com}}}
\newcommand{\txrev}{\tx_{\mathsf{rev}}}

\newcommand{\id}{i}
\newcommand{\commitlist}{\mathtt{c}}

\newcommand{\dealercon}{\texttt{Dealer}\xspace}
\newcommand{\scrapecon}{\texttt{Scrape}\xspace}
\newcommand{\caucuscon}{\texttt{Caucus}\xspace}

\newcommand\sarahm[1]{}
\newcommand\saraha[1]{}

\newcommand{\sysname}{Caucus\xspace}

\title{Winning the Caucus Race: Continuous Leader Election via Public Randomness}
\author{Sarah Azouvi, Patrick McCorry and Sarah Meiklejohn\\
\vspace{2mm}
University College London, United Kingdom}
\date{}

\begin{document}

\maketitle

\setlength\epigraphwidth{0.9\linewidth}
\setlength\epigraphrule{0pt}
\epigraph{
\emph{There was no `One, two, three, and away,' but they began running when they
liked, and left off when they liked, so that it was not easy to know when the
race was over.}
}{\dash Lewis Carroll, \emph{Alice's Adventures in Wonderland}\\
\emph{A Caucus-Race and a Long Tale}}

\begin{abstract}

Consensus protocols inherently rely on the notion of leader election, in
which one or a subset of participants are temporarily elected to authorize and
announce the network's latest state.  While leader election is a well studied 
problem, the rise of distributed ledgers (i.e., blockchains) has led to a new 
perspective on how to perform large-scale leader elections via solving a 
computationally difficult puzzle (i.e., proof of work). 
In this paper, we present \sysname, a large-scale leader election protocol 
with minimal coordination costs that does not require the computational cost
of proof-of-work.  We evaluate \sysname in terms of its security, using a new 
model for blockchain-focused leader election, before testing an 
implementation of \sysname on an Ethereum private network.  Our experiments 
highlight that one variant of \sysname costs only \$0.10 per leader election 
if deployed on Ethereum. 

\end{abstract}

\section{Introduction}


One of the central components of any distributed system is a \emph{consensus
protocol}, by which the system's participants can agree on its current
state and use that information to take various actions; e.g., decide whether
or not to accept a certain transaction as valid.  At heart, many consensus
protocols rely on a \emph{leader election} protocol, in which one
participant or subset of participants is chosen to lead these types of
decisions for a single round (or set period of time).
To elect a leader, participants must coordinate amongst themselves by
exchanging several rounds of messages. 

In settings such as those exemplified by distributed ledgers\dash and in
particular by open blockchains such as Bitcoin\dash it is not
clear how one could adopt traditional leader election protocols, as the
desired level of decentralization amongst participants makes 
them appear inappropriate. 
As a result, the consensus protocol most often used in this
setting is \emph{proof-of-work} (PoW), which in some sense replaces the
coordination costs of traditional protocols with the high computational cost
of finding the partial pre-image of a hash, or solving some other
cryptographic puzzle.  Even in this seemingly leaderless consensus protocol,
however, there is still an implicit notion of leadership, as the first 
participant to solve the puzzle effectively wins the authority to decide 
which transactions are accepted into the blockchain. 

One natural area to explore in the setting of distributed ledgers is
therefore the extent to which a compromise can be made between the
coordination costs of traditional leader election protocols and the
computational costs of PoW-based protocols, and indeed in recent years this
has been an active area of
research~\cite{CCS:MXCSS16,CCS:LNZBGS16,NDSS:DanMei16,iohk-pos,snowwhite,
algorand,casper}.
While some protocols attempt this compromise by inserting more traditional
trust models and consensus protocols into the decentralized setting, others
attempt to remove the computational cost of PoW without requiring any
coordination costs.  Instead, these consensus protocols replace the (explicit)
investment of computational cost in PoW-based mining with some (implicit)
represented investment in the health of the system on behalf of participants.
As such, they are often referred to as \emph{proof-of-stake} (PoS).

Given the lack of explicit computational cost, PoS-based protocols are
subject to new types of attacks, and therefore careful and detailed analysis
of both their leader election protocols\dash which, in contrast to PoW, must be
made explicit\dash and their incentive mechanisms.  While previous papers on
PoS often try to solve both of these problems at once, in this paper we focus
solely on the question of leader election.  By addressing just this well
defined question (`how to securely agree on a leader without having to
coordinate?'), we are able to model, analyze, and evaluate our proposed leader
election protocol in a more rigorous manner than would be possible if trying
to present a full PoS solution. We leave the incentives design for future work that requires
a paper on its own.

With this in mind, we begin in Section~\ref{sec:model} with a model for
leader election protocols that is inspired by existing models for coin-tossing
protocols but takes into account the unique requirements of a
blockchain-based setting.  In particular, we consider (1) \emph{liveness},
which is an
essential requirement for any distributed protocol; (2) \emph{fairness}, as
blockchains often rely on explicit incentives in the form of a reward when
elected leader, so it is important to consider how often this happens and if
it can be biased; and (3) different flavors of \emph{unpredictability}, as if
an adversary could predict either their own eligibility or that of another
participant then it could launch attacks accordingly (in particular a
\emph{grinding} or a DoS attack, respectively).

Next, in Section~\ref{sec:construction} we present \sysname, our leader
election protocol.  Briefly, \sysname works by combining a source of private
randomness (used to hide one's eligibility from other participants until it
is appropriate to reveal it) with a source of public randomness (used to
guarantee fairness).  To create the public randomness, we begin by looking to 
existing protocols for random beacons~\cite{randhound,iohk-pos,scrape}.  To 
create the private 
randomness, we simply allow each participant to do so themselves, and then 
ask them to publicly commit to it by revealing the head of a hash chain.  

Our experiments in Section~\ref{sec:impl} highlight that even the cheapest
random beacon we could find, SCRAPE~\cite{scrape}, costs over \$270 to 
generate in an on-chain Ethereum-based implementation, whereas a hybrid 
version that moves some computation off-chain costs under \$80.  This 
motivated us to create a random beacon inspired by
RanDAO~\cite{vitalik-randomness}, that combines both public and private
sources of randomness.  This new solution relies on a
somewhat stronger trust assumption (that at least one honest participant has
contributed to the beacon) but costs only \$0.10 per leader election.
We also discuss in Section~\ref{sec:improvements} how a leader
election protocol like \sysname could be used to bootstrap from a PoW- to a
PoS-based consensus protocol, although again we leave a detailed analysis of 
this problem and our solution as important future work, that requires an incentives model
and analysis.

In summary, our paper makes the following concrete contributions:

\begin{itemize}
\item We present the first model for leader election that is designed to
capture the specific security requirements associated with blockchains;

\item We present a leader election protocol that, in addition to satisfying
more traditional notions of security, ensures that leaders cannot be
subject to DoS attacks because their eligibility is revealed ahead of time;
and

\item We present Ethereum-based implementations of both our leader
election protocol and its underlying random beacon that are cheap enough to
be deployed today.
\end{itemize}

\section{Related Work}

The work that most closely resemble ours spans from leader election protocols
in the distributed computing literature to more recent proof-of-stake
proposals in the cryptography literature.

In the distributed systems literature, the leader election problem consists of
picking one player (or processor) out of $n$, considering that a fraction $b$
of them are \emph{bad} (faulty or malicious).  Most of the proposed solutions
work in the full information model~\cite{collectivecoinflipping}, where parties
are computationally unbounded and every communication is broadcast.  The
first leader election protocol is due to Saks~\cite{saks-leaderelection}, who
proposed a baton-passing algorithm resilient to a $\Theta(n/\log(n))$
threshold of bad players.  Since then, various protocols have improved on
either the communication complexity or fraction of malicious
players~\cite{alon-naor-election,orv-leaderelection,russel-leaderelection,
feige-leaderelection,king-leaderelection}.

In the cryptography literature, a number of recent papers have proposed
proof-of-stake protocols~\cite{iohk-pos,snowwhite,algorand}.
Briefly, some of the
techniques are similar, but many of the proposals either rely solely on a
source of global randomness (meaning they do not achieve as strong a notion
of unpredictability), or do not use any global randomness (meaning they do
not achieve as strong a notion of fairness).
In the Snow White proof-of-stake protocol~\cite{snowwhite}, the source of
randomness is taken as some $\nonce$ in a previous block in the blockchain,
and the leader is then any participant with public key $\pk$ such that
$H(\nonce\|\pk\|t)<\target$ (where $t$ is some agreed-upon time).  To ensure
that all participants agree on the value of $\nonce$ it must be sufficiently
far back in the blockchain (in case of forks), which means the protocol does
not satisfy any strong notion of unpredictability.  Furthermore, $\nonce$ is
chosen by the participant who created this previous block, so as this value
could have been biased by them it is also not possible to prove fairness.

The other main proposed protocol that does not explicitly rely on randomness
generation is Algorand~\cite{algorand}, whose solution in many ways resembles
the RanDAO-inspired approach we describe in Section~\ref{sec:randao}.  In this
protocol, a value $Q_0$ is initialized randomly, and then for every round
$\round$, a new value is defined as
$Q_\round\gets H(SIG_{l^r}(Q_{\round-1})\|\round)$, where $SIG_{l^r}$ is a
deterministic signature under the leader of round $\round$.  Algorand does not
specify, however, how $Q_0$ could be initialized, saying just that it is ``a
random number, part of the system description, and thus publicly known.''
This leaves open the possibility that whoever has set it may have again run a
grinding attack to bias the protocol (even if doing so would be
very computationally expensive).  Again then, it is not clear how to prove
fairness under the assumption that participants wouldn't act to bias the
protocol in their own favor (which seems economically rational).
\saraha{add more about Ouroboros praos}
Another similar solution is presented in Ouroboros Praos~\cite{praos}. In this protocol,
a random nonce is also initialized and then updated using a hash function applied to a Verifiable Random Function (VRF) using values from previous blocks.
However their beacon can be manipulated by the adversary as a number of ``resets'' can be performed
and thus it does not satisfy our security guarantees neither.
To the best of our knowledge
then, \sysname is the first implemented solution to achieve all security
guarantees.

Finally, a related notion to the idea of leader election is that of generating
public shared randomness~\cite{EPRINT:BonClaGol15,randhound,scrape}.  While
most existing protocols do not scale to hundreds or thousands of participants,
recent work such as RandHound~\cite{randhound} (and its related RandShare and
RandHerd protocols) work to achieve this goal.  In
Appendix~\ref{appendix:scrape}, we present the SCRAPE protocol~\cite{scrape},
due to Cascudo and David, that we incorporate into \sysname as a way to
initialize a random beacon.

\section{Background Definitions and Notation}\label{sec:defns}

In this section, we present the underlying notation and cryptographic
primitives we rely on in the rest of the paper.  In particular, we begin with
notation (Section~\ref{sec:notation}), and then present definitions of hash
functions (Section~\ref{sec:defns-hash}), and coin
tossing (Section~\ref{sec:defns-cointoss}).  We end with an overview of
blockchains (Section~\ref{sec:defns-blockchain}), and in particular give
definitions of proof-of-stake and some specific notation
associated with the Ethereum platform.

\subsection{Preliminaries}\label{sec:notation}

If $x$ is a binary string then $|x|$ denotes its bit length.  If $S$ is a
finite set then $|S|$ denotes its size and $x\randpick S$ denotes sampling a
member uniformly from $S$ and assigning it to $x$.  $\secp\in\N$ denotes the
security parameter and $\usecp$ denotes its unary representation.

Algorithms are randomized unless explicitly noted otherwise.  ``PT'' stands
for ``polynomial time.'' By $y\gets A(x_1,\ldots,x_n;R)$ we denote running
algorithm $A$ on inputs $x_1,\ldots,x_n$ and random coins $R$ and assigning
its output to $y$.  By $y\randpick A(x_1,\ldots,x_n)$ we denote $y\gets
A(x_1,\dots,x_n;R)$ for $R$ sampled uniformly at random.  By
$[A(x_1,\ldots,x_n)]$ we denote the set of values that have non-zero 
probability of being output by $A$ on inputs $x_1,\ldots,x_n$.  Adversaries 
are algorithms.  We denote non-interactive algorithms using the font
$\mathsf{Alg}$, and denote interactive protocols using the font
$\mathtt{Prot}$.  We further denote such protocols as
$\outputs\randpick\mathtt{Prot}(\usecp,\participants,\inputs)$, where
the $i$-th entry of $\inputs$ (respectively, $\outputs$) is used to denote 
the input to (respectively, output of) the $i$-th 
participant.

We say that two probability ensembles $X$ and $Y$ are statistically close over
a domain $D$ if
$\frac{1}{2}\sum_{\alpha\in D}|\pr[X=\alpha]-\pr[Y=\alpha]|$ is negligible; we
denote this as $X\approx Y$.  

\subsection{Hash functions}\label{sec:defns-hash}

A hash function is a function $H:\{0,1\}^*\rightarrow\{0,1\}^\ell$; i.e., a
function that maps strings of arbitrary length to strings of some fixed length
$\ell$.  When a hash function is modeled as a random oracle, this means that
computing $H(x)$ is modeled as (1) looking up $x$ in some global map and
using the value of $H(x)$ if it has been set already, and (2) if not, picking a
random value $y\randpick\{0,1\}^\ell$ and setting $H(x)\gets y$ in the map.

There are many desired properties of hash functions that trivially hold when
they are modeled as random oracles.  Of these, the two we are particularly
interested in \emph{unpredictability} and \emph{secrecy}~\cite{C:Canetti97}.
These are defined as follows:

\begin{definition}{\emph{\tbf{\cite{C:Canetti97}}}}
\label{def:hash-unpredictability}
A hash function satisfies \emph{unpredictability} if no PT adversary can
find an $x$ such that $(x,H(x))$ has some desired property.  More formally,
let $I_x$ be an oracle that, on queries $z$, returns whether or not $z = x$.
Then for all $x$, predicates $P(\cdot)$, and PT adversaries $\A$ that have
access to $H(\cdot)$, there exists an algorithm $\A'$ that has access only to
$I_x$ such that $\A'$ has roughly the same advantage in predicting $P(x)$ as
does $\A$.
\end{definition}

\begin{definition}{\emph{\tbf{\cite{C:Canetti97}}}}
\label{def:hash-secrecy}
A hash function satisfies \emph{secrecy} if, given $H(x)$, no PT adversary
with binary output can infer any information about $x$.  Formally, for a
domain $D$ and for all $x,y\in D$,
\[
\langle x,\A(H(x))\rangle \approx \langle x,\A(H(y))\rangle.
\]

\end{definition}

\subsection{Coin tossing and random beacons}\label{sec:defns-cointoss}

Coin tossing is closely related to leader
election~\cite{collectivecoinflipping}, and allows two or more parties to
agree on a single or many random bits~\cite{blum-cointossing,FOCS:BenLin85,EPRINT:Popov16};
i.e., to output a value $R$ that is statistically close to
random.

A coin-tossing protocol must satisfy \emph{liveness}, \emph{unpredictability},
and \emph{unbiasability}~\cite{randhound}, where we define these (in keeping with our
definitions for leader election in Section~\ref{sec:model}) as follows:

\begin{definition}\label{def:coin-liveness}
Let $\afrac$ be the fraction of participants controlled by an adversary $\A$.
Then a coin-tossing protocol satisfies \emph{$\afrac$-liveness} if it is still
possible to agree on a random value $R$ even in the face of such an $\A$.
\end{definition}

\begin{definition}\label{def:coin-unpredictability}
A coin-tossing protocol satisfies unpredictability if, prior to some step
$\barrier$ in the protocol, no PT adversary can produce better than a random
guess at the value of $R$.
\end{definition}

\begin{definition}\label{def:coin-unbiasability}
A coin-tossing protocol is \emph{$\afrac$-unbiasable} if for all PT
adversaries $\A$ controlling an $\afrac$ fraction of participants, the
output $R$ is still statistically close to a uniformly distributed random
string.
\end{definition}

Many coin-tossing protocols~\cite{blum-cointossing,EPRINT:Popov16} follow a
structure called
\emph{commit-then-reveal}: to start, in the commit phase each participant
creates and broadcasts a cryptographic commitment~\cite{Goldreich04} to a random
value.  In the reveal phase, participants broadcasts the opening
of their commitments, and can use the opening to check the validity of the
initial commitment.
The output value is then some combination (e.g., XOR) of all the individual
random values.  Intuitively, while this basic solution does not satisfy
liveness, it satisfies unpredictability due to the hiding
property of the commitment scheme (where $\barrier$ corresponds to the step in
which the last commitment is opened), and unbiasability due to the binding
property.
In order to achieve liveness, \emph{secret sharing} can be used.  To ensure
that the protocol produces a valid output, participants create shares of
their secret and send these shares to other participants during the commit
phase. This allows participants to recover the value of another participant
even if they abort.

A closely related concept to coin tossing is \emph{random beacons}.  These
were first introduced by Rabin~\cite{rab81} as a service for ``emitting at
regularly spaced time intervals, randomly chosen integers''.  To extend the
above definitions to random beacons, as inspired by~\cite{EPRINT:BonClaGol15}, we require that
the properties of $\afrac$-liveness and $\afrac$-unbiasability apply
for each iteration of the beacon, or \emph{round}.  We also require that the
$\barrier$ in the unpredictability definition is at least the
begining of each round.  We present such a scheme, SCRAPE~\cite{scrape}, in
more detail in Appendix~\ref{appendix:scrape}.

\subsection{Blockchain basics}\label{sec:defns-blockchain}

Distributed ledgers, or blockchains, have become increasingly popular ever
since Bitcoin was first proposed by Satoshi Nakamoto in
2008~\cite{nakamoto-bitcoin}.
Briefly, individual Bitcoin users wishing to pay other users broadcast
\emph{transactions} to a global peer-to-peer network, and the peers
responsible for participating in Bitcoin's consensus protocol (i.e., for
deciding on a canonical ordering of transactions) are known as \emph{miners}.
In order to participate, miners form \emph{blocks}, which contain (among
other things that we ignore for ease of exposition) the transactions collected
since the
previous block was mined, which we denote $\txset$, a pointer to the previous
block hash $h_{\mathsf{prev}}$, and a \emph{proof-of-work} (PoW).  This PoW
is the solution to a computational puzzle specified by a value $\target$;
roughly, it is a value $\nonce$ such that $H(\nonce\|\txset\|h_{\mathsf{prev}})
< \target$, and indicates that the miner has put in the computational effort
(or work) necessary to find $\nonce$.  Every other peer in the network can
verify the validity of a block, and if future miners choose to include it by
incorporating its hash into their own PoW then it becomes part of the global
\emph{blockchain}.
In the rare case where two miners find a block at the same time,
other miners decide which block to mine off of.  Bitcoin follows the
\emph{longest chain} rule, meaning that whichever fork creates the longest
chain is the one that is considered valid.  The block that hasn't been
included in that fork is thus abandoned.

\subsubsection{Proof-of-stake}\label{sec:defns-pos}

By its nature, PoW consumes a lot of energy.  Thus, some alternative consensus
protocols have been proposed that are more cost-effective; as of this writing,
arguably the most popular of these is called \emph{proof-of-stake}
(PoS)~\cite{pos-forum,ppcoin,posaltcoins,casper}. If
we consider PoW to be a leader election protocol in which the leader (i.e.,
the miner with the valid block) is selected in proportion to their amount of
computational power, then PoS can be seen as a leader election protocol in
which the leader (i.e., the participant who is \emph{eligible} to propose a
new block) is selected in proportion to some ``stake'' they have in the
system.  This can be represented as the amount of coins they have (either in
some form of escrow or just in total), or the age of their coins.

As security no longer stems from the fact that it is expensive to create a
block, PoS poses several technical challenges~\cite{interactivepos}.  The
main three are as follows: first, the \emph{nothing at stake} problem says
that have no reason to not mine on top of every chain, since mining is
costless, so it is more difficult to reach consensus.  This is an issue of
incentives, so given our focus on leader election we consider it out of scope
for this paper.

Second, PoS allows for \emph{grinding attacks}, in which once a miner is
elected leader they privately iterate through many valid blocks (again,
because mining is costless) in an attempt to find one that may give them an
unfair advantage in the future (e.g., make them more likely to be elected
leader).  This is an issue that we encounter in Section~\ref{sec:construction}
and address using a public source of randomness.

Finally, in a \emph{long range attack}, an attacker may bribe miners into
selling their old private keys, which would allow them to re-write the entire
history of the blockchain.  Again, this is an issue of incentives that we
consider out of the scope of this paper.

\subsubsection{Smart contracts and Ethereum}

Bitcoin's scripting language is intentionally limited, and intended to support
mainly the transfer of coins from one set of parties to another.  In contrast,
the Ethereum platform provides a more complex scripting language that allows
it to act as a sort of distributed virtual machine, and to support more
complex objects called \emph{smart contracts}.

Ethereum transactions contain a destination address $\dest$ (that can specify
either the location of a stateful smart contract, or just a regular user-owned
address), an amount $\amount$ to be sent (denominated in ether), an optional
data field $\data$, a \emph{gas limit}, and a signature $\sigma$ authorizing
the transaction with respect to the sender's keypair $(\pk,\sk)$ (just as in
Bitcoin).  We ignore everything except the $\data$ field in our protocol
specifications below and in Section~\ref{sec:construction}, but mention
briefly here that gas is a subcurrency in Ethereum used to pay for the
computational operations performed by a smart contract.
We denote the process of creating a transaction as (again, ignoring all but
the $\data$ field) $\tx\randpick\formtx(\sk,\data)$, and the process of
verifying a transaction as $0/1\gets\verifytx(\tx)$.

\section{A Model for Leader Election}\label{sec:model}

Most of the consensus protocols in the distributed systems literature
are leader-based, as it is the most optimal solution in term of
coordination~\cite{leaderless-consensus}.  Perhaps as a result, leader
election has in general been very well studied within the distributed systems
community~\cite{saks-leaderelection,russel-leaderelection,feige-leaderelection,orv-leaderelection,king-leaderelection}.

Nevertheless, to the best of our knowledge the problem of leader election has
not been given an extensive cryptographic or security-focused treatment, so in
this section we provide a threat model in which we consider a variety of
adversarial behavior.  This is particularly crucial in open
applications like consensus protocols for blockchains, as any actor can
participate in the
consensus protocol, so we cannot place any trust in the set of potential
leaders or assume any built-in Sybil resistance.  While this is our motivating
application, we consider this type of leader election to be useful in any
application in which some limited set of untrusted participants should be
eligible to take some action.

\subsection{The setting}

We consider a round-based leader election protocol run between a set
of participants $\participants$; that is, a protocol in which the participants
want to agree on who is eligible to take some action (e.g., propose a new
block) in some future round $\round$.

Each participant maintains some private state
$\privatestate$, and all participants are expected to agree on some
public state $\pubstate$.  For ease of exposition, we assume each
$\privatestate$ includes the public state $\pubstate$.  We refer to a
message sent by a participant as a \emph{transaction}, denoted $\tx$, where
this transaction can either be broadcast to other participants (as in a more
classical consensus protocol) or committed to a public blockchain.

Our model is inspired by the classical approach of coin-tossing protocols
that proceed in a commit-then-reveal fashion.  It consists of four algorithms
and one interactive protocol, which behave as follows:

\begin{description}

\item[$\pubstate\randpick\setup(\usecp)$] is used to establish the
initial public state $\pubstate$.

\item[$(\privatestate,\txdep)\randpick\commit(\pubstate,
\round_\maxr)$] is used by a participant to commit themselves to participating
in the leader election up to some round $\round_\maxr$.  This involves
establishing both an initial private state $\privatestate$ and a public
announcement $\txdep$.

\item[$\{\privatestate^{(i)}\}_i\randpick\coord(\usecp,\participants,
\{(\round,\privatestate^{(i)})\}_i)$] is run amongst the committed
participants, each of whom is given $\round$ and their own
private state $\privatestate^{(i)}$, in order to update both the
public state $\pubstate$ and their own private states to prepare the leader
election for round $\round$.

\item[$\txrev\randpick\reveal(\round,\privatestate)$] is used by
a participant to broadcast a proof of their eligibility $\txrev$ for round
$\round$ (or $\bot$ if they are not eligible).

\item[$0/1\gets\verify(\round,\pubstate,\txrev)$] is used by a participant
to verify a claim $\txrev$ of eligibility for round $\round$.

\end{description}

\subsection{Security properties}

We would like a leader election protocol to achieve three security properties:
\emph{liveness}, \emph{unpredictability}, and \emph{fairness}.  The first
property map to the established property of liveness for
classical consensus protocols, although as we see below we consider several
different flavors of unpredictability that are specific to our motivating
blockchain-based application.  The final one, fairness (also called
\emph{chain quality}~\cite{iohk-pos}), is especially important in open
protocols like blockchains, in which participation must be explicitly
incentivized rather than assumed.

We begin by defining liveness, which requires that consensus can be achieved
even if some fraction of participants are malicious or inactive.

\begin{definition}[Liveness]\label{def:liveness}
Let $\afrac$ be the fraction of participants controlled by an adversary $\A$.
Then a leader election protocol satisfies \emph{$\afrac$-liveness} if it is
still possible to elect a leader even in the face of such an $\A$; i.e., if
for every public state $\pubstate$ that has been produced via
$\coord$ with the possible participation of $\A$, it is still possible for at
least one participant, in a round $\round$, to output a value $\txrev$ such that
$\verify(\round,\pubstate,\txrev)=1$.
\end{definition}

Next, unpredictability requires that participants cannot predict which other
participants will be elected leader before some point in time.

\begin{definition}[Unpredictability]
A leader election protocol satisfies \emph{unpredictability} if, prior to some
step $\barrier$ in the protocol, no PT adversary $\A$ can produce better than
a random guess at
whether or not a given participant will be eligible for round $\round$.  If
$\barrier$ is the step in which a participant broadcasts $\txrev$, and we
require $\A$ to guess only about the eligibility of honest
participants (rather than participants they control), then we say it
satisfies \emph{delayed} unpredictability.  If it is still difficult for $\A$
to guess even about their own eligibility, we say it satisfies
\emph{private} unpredictability.
\end{definition}

Most consensus protocols satisfy only the regular variant of unpredictability
we define, where $\barrier$ is the point at which the $\coord$ interaction is
``ready'' for round $\round$ (e.g., the participants have completed a
coin-tossing or other randomness-generating protocol).  This typically occurs
at the start of the round, but may also occur several rounds before it.

If an adversary is aware of the eligibility of other participants ahead of
time, then it may be able to target these specific participants for a
denial-of-service (DoS) attack, which makes achieving liveness more difficult.
A protocol that satisfies delayed
unpredictability solves this issue, however, as participants reveal their
own eligibility only when they choose to do so, by which point it may be too
late for the adversary to do anything. (For example, in a proof-of-stake
protocol, if participants include proofs of eligibility only in the blocks
they propose, then by the time the leader is known the adversary has nothing
to gain by targeting them for a DoS attack.)

A protocol that satisfies private unpredictability, in contrast, is able to
prevent an adversary from inflating their own role as a leader.  For
example, if an adversary can predict many rounds into the future what their
own eligibility will be, they may attempt to bias the protocol in their favor
by \emph{grinding}~\cite{interactivepos} through some problem space in order to produce an
initial commitment $\txdep$ that yields good future results.

Private unpredictability thus helps to guarantee fairness, which we define as
requiring that each committed participant is selected as leader equally often.
While for the sake of simplicity our definition considers equal weighting of
participants, it can easily be extended to consider participants with respect
to some other distribution (e.g., in a proof-of-stake application, participants
may be selected as leader in proportion to their represented ``stake'' in the
system).

\begin{definition}[Fairness]
A leader election protocol is \emph{$\afrac$-fair} if for all PT adversaries
$\A$ controlling an $\afrac$ fraction of participants,
the probability that $\A$ is selected as leader is nearly uniform;
i.e., for all $\round$, $\privatestate$, $\pubstate$ (where again
$\pubstate$ has been produced by $\coord$ with the possible participation of
$\A$), and $\txrev$ created by $\A$,
$
\pr[\verify(\round,\pubstate,\txrev)=1]\approx \afrac.
$
\end{definition}

\section{\sysname: A Leader Election Protocol}\label{sec:construction}

In this section, we present \sysname, a leader election protocol with
minimal coordination that satisfies fairness, liveness, and strong notions of
unpredictability.  In Section~\ref{sec:impl}, we present and evaluate an
implementation of \sysname using the Ethereum blockchain.

To begin, we present a straw-man solution in Section~\ref{sec:straw-man}.  The
benefits of this solution are that it in fact requires no coordination
whatsoever; participants are elected leader based on privately held (but
publicly committed to) beliefs about their own eligibility, which they then
reveal to other participants when relevant.  While this means the protocol
achieves liveness and the strong notion of delayed unpredictability, it cannot
achieve fairness due to participants' reliance on their own personal
randomness, which they can bias without significant effort.

To address this shortcoming, in Section~\ref{sec:randomness} we discuss how to
fold a source of shared randomness into our protocol in order to make it fair.

Unfortunately, even the best of these solutions still requires significant
coordination, which is prohibitively expensive in a setting in which we want
leader election to occur frequently (e.g., every time a new block is
produced).  Thus, in Section~\ref{sec:randao} we explore how to bootstrap the
guarantees of these protocols to extend the randomness in a cheaper
way.  Here our techniques are inspired by the RanDAO protocol described by
Buterin~\cite{vitalik-randomness}.

Finally, we put everything together and present in
Section~\ref{sec:sysname} our full solution, \sysname, which combines
the private randomness of our straw-man solution, in order to achieve
delayed unpredictability, with a cheap source of public randomness, in
order to achieve fairness and private unpredictability.

\subsection{A straw-man solution}\label{sec:straw-man}

We begin with a simple straw-man solution in which participants commit to some
privately held randomness by forming a hash chain; i.e., a list
$(h_1,\ldots,h_n)$ such that $h_1 = H(\seed)$ for some random seed and $h_i =
H(h_{i-1})$ for all $i$, $2\leq i\leq n$.  We denote by $H^i(\seed)$ the
application of $H$ to $\seed$ for $i$ times.

In $\setup$, the public state is just initialized to be an empty list;
eventually, it will contain a list of commitments $\commitlist$ posted by
participants, in the form of the heads of their hash chains.  For this and all
of our subsequent solutions, we assume participants have generated signing
keypairs $(\pk,\sk)$ and are aware of the public key associated with each
other participant (which can easily be achieved at the time participants
run $\deposit$).  We omit the process of generating keys from our formal
descriptions but denote by $\tx[\id]$ the participant associated with a
given transaction (which corresponds to the public key used to sign it).

To to be considered as potential leaders, participants must place a
\emph{security deposit}, that involves creating a commitment to their
randomness.  In our implementation in Section~\ref{sec:impl}, this entails
including the commitment in a transaction that also puts into escrow some
monetary value and sending it to an Ethereum smart contract, but in general
the commitment could just be a message broadcast to all other participants.
This means the $\deposit$ function picks a random seed, forms a hash chain
of length $\round_\maxr$, and returns the seed itself as the private state
of the participant and the head of the hash chain (incorporated into a
transaction) as the transaction to add to the public state.  More formally:

\begin{center}
\begin{tabular}{l}
\underline{$\deposit(\pubstate,\round_\maxr)$}\\
$\seed\randpick\{0,1\}^*$\\
$\txdep\randpick\formtx(\sk,H^{\round_\maxr}(\seed))$\\
return ($\seed,\txdep$)
\end{tabular}
\end{center}

The eligibility of each participant to act as leader is then based purely on
the privately held beliefs contained inside their hash chain.  In particular,
this solution requires no coordination and thus no $\coord$ protocol.
Instead, in each round $\round$ a participant can check if they are
eligible by checking if $H^{\round_\maxr-\round}(\seed)<\target_\round$ for
$\target_\round = \hashmax / n_\round$ (i.e., the maximum hash value divided
by the number of potential leaders in $\round$).  They can then prove their
eligibility by peeling back to that layer of their hash chain, which means
running $\reveal$ as follows:

\begin{center}
\begin{tabular}{l}
\underline{$\reveal(\round,\privatestate)$}\\
$h_\round\gets H^{\round_\maxr-\round}(\seed)$\\
return $\txrev\gets\formtx(\sk,h_\round)$\\
\end{tabular}
\end{center}

As with $\deposit$, in a more general system this could involve simply
broadcasting the value $h_\round$ itself, rather than forming it into a
transaction.

Finally, once a $\reveal$ transaction has been broadcast, other participants
can verify the eligibility of that participant by checking that the layer
revealed correlates with the commitment to their hash chain.  More formally:

\begin{center}
\begin{tabular}{l}
\underline{$\verify(\round,\commitlist=\pubstate,\txrev)$}\\
$h_\round\gets\txrev[\data]$\\
$\id\gets\txrev[\id]$\\
return $(h_\round < \target_\round)\land (H^\round(h_\round) = \commitlist[\id])$
\end{tabular}
\end{center}

As stated, participants can join by broadcasting a $\commit$ transaction at
any time after the protocol has started.  To maintain unpredictability,
however, participants should become eligible only after some fixed number of
rounds have passed since they ran $\commit$.  In this case, $\verify$ would
check not only that $h_\round < \target_\round$, but also that
$\round_\mathsf{joined} > \round - x$ (where $\round_\mathsf{joined}$ is the
round in which the participant broadcast $\txdep$ and $x$ is the required
number of interim rounds).

Because of its lack of coordination, and as long as the commitment hides the
previous values in the hash chain, this protocol clearly satisfies delayed
unpredictability, as participants never talk to each other before they reveal
their own eligibility.  It also achieves a notion of liveness, as again
participants do not rely on each other to establish their own eligibility.  It
could be the case in some round, however, that no participant is elected
leader.  In this case, to maintain liveness we consider that if after a
small delay no participants have broadcast a $\txrev$ transaction, we skip
the round and participants update the value of $\round$ as
$\round\gets\round+1$. This ensure that the protocol keeps
running and that a participant will be elected in some future round.

\begin{lemma}\label{lem:strawman-unpredictable}
If $H$ is a random oracle, then the protocol described above satisfies
liveness and delayed unpredictability.
\end{lemma}

\begin{proof}
The protocol satisfies liveness, as no coordination is required
amongst participants.  In particular then, any online participant can
communicate their own eligibility to other online participants and be
considered leader.  The exception is if no online participants satisfy the
condition that $h_\round < \target_\round$, but then we can either
skip this round as explained above, or modify
$\target_\round$ to ensure that this never happens.

To satisfy delayed unpredictability, we must show that\dash unless it has
formed $h_\round$ itself\dash no adversary $\A$ can predict whether or not
$h_\round < \target_\round$ before the point at which it is given $\txrev$,
and thus in particular before it is given $h_\round$
itself.  In the protocol, the adversary sees $H^{\round_\maxr}(\seed)$ as part
of the commitment of the relevant honest participant, and if that participant
has run $\reveal$ before it may have also seen $H^{\round_\maxr-\round'}(\seed)$
for $\round' < \round$.  In the worst case, it has seen this for $\round' =
\round - 1$, meaning it has seen $H^{\round_\maxr - \round + 1}(\seed) =
H(h_\round)$.

If we define $x = h_\round$, then this is equivalent to giving $H(x)$ to $\A$
and asking it to predict some predicate $P(x)$, where in this case
$P(z) = 1$ if $z < \target_\round$ and $0$ otherwise.  By the unpredictability
of random oracles (Definition~\ref{def:hash-unpredictability}), we know that
this is not possible.
\end{proof}

If we use $\target_\round$ as described above, then in theory our protocol
should also be able to establish fairness if $\seed$ is uniformly distributed.
Indeed, by setting $\target_\round$ on a per-participant basis, we could even
satisfy more complex fairness guarantees, such as ensuring that participants
who have put more money into escrow in their security deposit get chosen as
leader more often.

As $\seed$ is chosen by each participant rather than uniformly at random,
however, we cannot actually
guarantee fairness.  In particular, participants can \emph{grind} through
different seeds in order to find one that unfairly favors them.  This is
due partially to the fact that the protocol is not privately unpredictable
with respect to any well defined barrier: once the adversary has decided on a
seed, it can predict its own eligibility in all future rounds (modulo the
shifting number of participants, which it can either grind through as well, or
wait until the number of participants is established before grinding).

\begin{lemma}\label{lem:strawman-notfair}
If $H$ is a random oracle, then the protocol described above is not fair.
\end{lemma}

\begin{proof}
As $H$ is a random oracle, every element in the hash chain can be modeled as
a random variable uniformly distributed on $[0, 2^\ell-1]$.  For a given
seed $\seed$, the probability $p$ for a given $\round$ that $h_\round <
\target_\round$ is thus $\target_\round / (2^\ell-1)$.

Define $M(\seed)$ to be the number of hashes $h_\round$ in a hash chain of
length $\round_\maxr$ where $h_\round < \target_\round$; i.e.,
$M(\seed) = \sum_{h_\round<\target}1$.  $M$ is thus a sum of Bernoulli random
variables, meaning it follows a Binomial distribution $B(\round_\maxr,p)$.
It it possible for the adversary to \emph{grind} (i.e., iterate) through
different seeds in order to find one that performs better.

In particular, an adversary can break fairness as follows: it tries many
different seeds in order to find one for which
$M>\round_\maxr p+\epsilon$ for some $\epsilon>0$ (with $\trm{average}(M)=\round_\maxr p$).
By the properties of Binomial distribution, we have that

\begin{align*}
\pr[M>k] &= 1-\pr[M\le k]\\
&= 1-\sum_{i=0}^k\binom{\round_\maxr}{i}p^i(1-p)^{\round_\maxr-i}.
\end{align*}

If we define this last term to be $\alpha$, then the expectation on the
number of trials needed for the adversary to find such a seed is
$1/\alpha$.  In each trial, the adversary computes $\round_\maxr$
hashes, so in total it needs to compute $\round_\maxr / \alpha$ hashes.
\end{proof}

\subsection{Folding in randomness}\label{sec:randomness}

Given the limitation of our straw-man solution that adversaries can grind
through random seeds, and thus fairness does not hold, we must consider ways
to improve this guarantee.

We begin by considering other solutions that do not have randomness, or that
rely on some source of randomness produced without an explicit coin-tossing
protocol.  First, as suggested by Buterin~\cite{vitalik-randomness}, we
consider using $H(\id)\oplus h_\round$ rather than just using $h_\round$.  While
this improves matters, as now an adversary controlling several different
participants would have to grind through separate seeds for each of them (as
opposed to finding one good seed and reusing it across all of them), a
grinding attack is still possible so the protocol is not fair.

In any solution where the source of randomness is purely private, it
seems somewhat impossible to prevent attacks on fairness, as adversaries can
always privately work to bias their own randomness, whether in the form of
privately held beliefs about their own eligibility or public randomness that
they produce single-handedly.  Furthermore, if the source of randomness were
known too far in advance (such as if we tried to use something like $H(0)$,
which is uniformly random if we model $H(\cdot)$ as a random oracle), then the
value would be known too far ahead of time to achieve unpredictability, in
which\dash as we see in the proof of Theorem~\ref{thm:plus-randomness}\dash
participants must commit before knowing the randomness in a given round.

To obtain meaningful notions of fairness and unpredictability, it thus seems
necessary to incorporate a source of public, globally sourced randomness.  To
augment our straw-man solution using such randomness, we first consider an
interactive
protocol $\coord$ that acts to establish some shared randomness $R_\round$ for
a given round.  This randomness would then be folded into the straw-man
solution by requiring not that $h_\round < \target_\round$ but that
$H(h_\round\oplus R_\round) < \target_\round$.
This results in the following changes:

\begin{center}
\begin{tabular}{l}
\underline{$\prove(\round,\privatestate)$}\\
if $(H(h_\round\oplus R_\round) < \target_\round)$ return $\formtx(\sk,h_\round)$\\
~\\
\underline{$\verify(\round,\pubstate,\txrev)$}\\
return $(H(h_\round\oplus R_\round) < \target_\round)\land (H^\round(h_\round) =
\commitlist[\id])$
\end{tabular}
\end{center}

In order to maintain liveness, we consider that if no participant has
revealed $\txrev$ before some delay, we update the random beacon as
$\globalrandao_\round\gets H(\globalrandao_\round)$.  If $\coord$ is a secure
random beacon, then we have the following theorem:

\begin{theorem}\label{thm:plus-randomness}
If $H$ is a random oracle and $R_\round$ was produced by a random
beacon satisfying liveness (Definition~\ref{def:coin-liveness}),
unpredictability (Definition~\ref{def:coin-unpredictability}), and
unbiasability (Definition~\ref{def:coin-unbiasability}), then this protocol
satisfies liveness, fairness, delayed unpredictability (i.e., where
$\barrier$ is the step at which the elected leader reveals their proof), and
private unpredictability (where $\barrier$ is the step at which the
randomness $\globalrandao_\round$ is fixed).
\end{theorem}

\begin{proof}
For liveness, a participant is elected if they broadcast a valid transaction
$\txrev$, such that $H(h_\round\oplus R_\round)<\target$. If $\coord$
satisifies liveness then an adversary controlling $\afrac$ participants
cannot prevent honest participants from agreeing on $R_\round$.
In the case where no participants produce a value $h_\round$ such that
$H(h_\round\oplus R_\round)<\target$, we update the value of $R_\round$ as
described above until one participant is elected.  With a similar argument as
in the proof of Lemma~\ref{lem:strawman-unpredictable}, the protocol thus
achieve liveness.

For fairness, a participant wins if $H(h_\round\oplus R_\round) <
\target_\round$.  By the assumption that $R_\round$ is unpredictable, and
thus unknown at the time an adversary commits to $h_\round$, the distribution
of the two values is independent.  Combining this with the assumption that
$R_\round$ is unbiasable and thus uniformly distributed, we can argue that
$h_\round\oplus R_\round$ is also uniformly distributed.  This implies that
the probability that the winning condition holds is also uniformly random, as
desired.

The argument for delayed unpredictability is almost identical to the one in
the proof of
Lemma~\ref{lem:strawman-unpredictable}: even when $R_\round$ is known, if $\A$
has not formed $h_\round$ itself then by the unpredictability of the random
oracle it cannot predict the value of $P(z)$, where $P(z) = 1$ if $H(z\oplus
R_\round) < \target_\round$ and $0$ otherwise.

Finally, private unpredictability follows from the unpredictability of
$R_\round$ and from the independence of $h_\round$ and $R_\round$.
\end{proof}

To be compatible with our implementation in Section~\ref{sec:impl}, the
concrete random beacon we use is the SCRAPE protocol, due to Cascudo and
David~\cite{scrape}.  For completeness, we sketch the original protocol and
present our modifications to it in Appendix~\ref{appendix:scrape}.  While we
chose SCRAPE due to its low computational complexity and compatibility with
public ledgers, several other protocols would be suitable as well, such as
RandShare or RandHound~\cite{randhound}, or the randomness generation process
used in Ouroboros~\cite{iohk-pos}.  Indeed, all of these solutions are similar
in that they instantiate a publicly verifiable secret sharing (PVSS) scheme
(also defined in Appendix~\ref{appendix:scrape}).


\subsection{Extending an initial random beacon}\label{sec:randao}

While any random beacon protocol provides the necessary security guarantees
for Theorem~\ref{thm:plus-randomness}, all existing protocols are relatively
expensive in terms of both local computations and, more crucially,
coordination costs.  This is not desirable in a setting such
as proof-of-stake, in which leader election must be run frequently and
continuously.

To minimize these costs, we therefore consider a significantly cheaper
solution in which the random value is initialized using a secure coin-tossing
protocol, but is then extended and updated using an approach inspired by
RanDAO~\cite{vitalik-randomness}.

In this protocol, participants place security deposits just as they did in our
straw-man solution, meaning they run the exact same $\deposit$ algorithm, and
each commitment is added to the public list $\commitlist$.

To initialize the random value, participants now run a coin-tossing protocol
to generate a random value $\globalrandao_1$ (in our implementation in
Section~\ref{sec:impl} we use SCRAPE, but again any secure coin-tossing
protocol could work).

For each subsequent round, participants then verify whether or not they
are eligible to fold their randomness into the global value by
checking if $H(\randao_\round \oplus \globalrandao_{\round})
< \target$, where $\randao_\round=H^{\round_\maxr-\round}(\seed)$.  If this
value holds, then they reveal their additional randomness by running the same
$\reveal$ algorithm as in our straw-man solution, which others can verify in
the same way as well.  If the participant is deemed to in fact be eligible,
then the global randomness is updated as $\globalrandao_{\round+1}\gets
\globalrandao_{\round} \oplus \randao_\round$.
In this setting, however, an adversary that controls $\afrac$ participants
could get lucky, privately predict that they will be elected leader for a
few rounds, and choose their new seed value accordingly.
To thwart this, we can again make sure participants have to wait some number
of rounds after they committed.  This raises the probability that an honest
participants will be elected leader between the time an adversary commits
and the time they are allowed to participate, thus making the random beacon
truly unpredictable and grinding attacks ineffective.

In a general system, each participant would need to maintain their own local
copy of $\globalrandao$ and update it appropriately.  In a system such as the
Ethereum blockchain (where RanDAO was originally proposed for use), a smart
contract can act to verify the claims of participants and fold the randomness
in by itself.  This makes the contract the authoritative source of randomness,
and because it is on the blockchain it is also globally visible.

\begin{theorem}\label{thm:randao}
If $H$ is a random oracle and $\globalrandao$ is initialized using a secure
coin-tossing protocol, then for every subsequent round this random beacon is
also secure; i.e., it satisfies liveness (Definition~\ref{def:coin-liveness}),
unbiasability (Definition~\ref{def:coin-unbiasability}), and unpredictability
(Definition~\ref{def:coin-unpredictability}) for every subsequent round.
\end{theorem}

\begin{proof}
For liveness, the argument is the same as in the proof of
Lemma~\ref{lem:strawman-unpredictable}: after initialization, no coordination
is required, so any online participant can communicate their own eligibility
to other online participants, allowing them to compute the new random value.
In the case where no participant is elected leader after some delay, then we
consider that participants update their random value by
$R_\round\gets H(R_\round)$ until a leader is elected.

For fairness, we proceed inductively.  By assumption, $\globalrandao$ is
initialized in a fair way, which establishes the base case.  Now, we assume
that $\globalrandao_{\round-1}$ is uniformly distributed, and
would like to show that $\globalrandao_\round$ will be as well.  By the same
argument as in the proof of Theorem~\ref{thm:plus-randomness}, by the
unpredictability of $R_\round$ and the fact that an adversary commits to
$h_{\round-1}$ before $\round-1$, the value of $h_{\round-1}$ is independent of the
value of $R_{\round-1}$ and can be considered as a constant.

If we
define a value $R$, and
denote $R'\gets R\oplus h_{\round-1}$, then we have

\begin{align*}
\pr[R_\round = R] &= \pr[h_{\round-1}\oplus R_{\round-1} = h_{\round-1}\oplus R']\\
                   &= \pr[R_{\round-1} = R'],
\end{align*}
which we know to be uniformly random by assumption, thus for every $R'$, we have
$\pr[R_{\round-1} = R']\approx 1/ (2^\ell-1)$ and for every value
$R$, $\pr[R_\round = R]\approx 1/ (2^\ell-1)$, proving $R_\round$ fairness.
(Here $\ell$ denotes the bitlength of $R$.)

\sarahm{What exactly does this prove though? It honestly seems like a bit of a tautology.}

For unpredictability, we must show that, unless the adversary is
itself the next leader, it is hard to learn the value of
$\globalrandao_\round$ before it receives $\txrev$.  We have
$\globalrandao_\round=\globalrandao_{\round-1}\oplus\randao_{\round-1}$,
where $\globalrandao_{\round-1}$ is assumed to be known.  By the same argument
about the unpredictability of $\randao_{\round-1}$ as in
Lemma~\ref{lem:strawman-unpredictable}, however, if it belongs to an honest
participant and $H$ is modeled as a random oracle, then the adversary is
unable to predict this value and thus unable to predict
$\globalrandao_\round$.
\end{proof}

Even if the initial value was some constant (e.g., 0) instead of a randomly
generated string, we could still argue that the protocol is fair after the
point that the randomness of at least one honest participant is incorporated
into the beacon. Indeed in this case the value of the beacon would be updated as
$\globalrandao_{\round+1}\gets\globalrandao_{\round}\oplus h_\round$ where
$h_\round$ is random since it belongs to an honest participant, and independent
of $\globalrandao_{\round}$, since it was committed to before round $\round$,
so this new value is random.  This assumption is weakened the longer the
beacon is live, so works especially well in settings such as we describe in
Section~\ref{sec:bootstrap}, in which the leader election protocol is used to
bootstrap from one form of consensus (e.g., PoW) to another (e.g., PoS).

\subsection{Our construction: \sysname}\label{sec:sysname}

Now that we have a cheap source of continuous randomness, we combine it with
the privately held beliefs in our straw-man solution in
Section~\ref{sec:straw-man} to give our full solution, \sysname.  The protocol
is summarized in Figure~\ref{fig:construction}.

\begin{figure*}[t]
\begin{framed}
\centering
\begin{description}[itemsep=5pt]
\item[$\setup$:]  In the $\setup$ phase, on input the security parameter
$\usecp$, participants initialize the initial public state of the protocol
$\pubstate$ with an empty list of deposits $\commitlist$.

\item [$\deposit$:] To participate up to round $\round_\maxr$, a participant
commits to a seed $\seed$ by creating a $\deposit$ transaction $\txdep$ that
contains a hash value $H^{\round_\maxr}(\seed)$.  Each broadcast commitment
is added to the list $\commitlist$ maintained in $\pubstate$, and that
participant is considered eligible to be elected leader after some fixed
number of rounds have passed.

\item [$\coord$:] Once enough participants are committed, participants run
a secure coin-tossing protocol to obtain a random value $R_1$.  They output a
new $\pubstate=(\commitlist, R_1)$.  \textbf{This interactive protocol is run
only for $\round = 1$.}

\item [$\reveal$:] For $\round > 1$, every participant verifies their own
eligibility by checking if $H(\globalrandao_\round\oplus \randao_\round)<\target$,
where $\randao_\round=H^{\round_\maxr-\round}(\seed)$ and
$\target=\hashmax/\participantsnumber_\round$. (Here
$\participantsnumber_\round$ is the number of eligible participants; i.e., the
number of participants that have committed a sufficient number of rounds before
$\round$.)  The eligible participant, if one
exists, then creates a transaction $\txrev$ with their data
$\randao_\round$ and broadcasts it to their peers.

\item [$\verify$:] Upon receiving a transaction $\txrev$, participants
extract $h_\round$ from $\txrev$ and check whether or not
$H(\globalrandao_\round\oplus \randao_\round)<\target$ and if $h_\round$ matches the
value committed; i.e., if $H^{\round}(h_\round) =
\commitlist[i]$.  If these checks pass,
then the public randomness is updated as $\globalrandao_{\round+1}\gets
\globalrandao_\round\oplus\randao_\round$ and they output $1$, and otherwise
the public state stays the same and they output $0$.
\end{description}
\end{framed}
\caption{Our \sysname protocol.}
\label{fig:construction}
\end{figure*}

\begin{theorem}\label{thm:caucus}
If $H$ is a random oracle and $R$ is initialized as a uniformly random value,
then \sysname is a secure leader election protocol; i.e., it satisfies
liveness, fairness, delayed unpredictability (i.e., where
$\barrier$ is the step at which the elected leader reveals their proof), and
private unpredictability (where $\barrier$ is the step at which the
randomness $\globalrandao_\round$ is fixed).
\end{theorem}

\begin{proof}
This theorem holds by the security of the components of the system, and in
particular because the security of the random beacon
(Theorem~\ref{thm:randao}) implies that the necessary conditions for the
security of the randomized solution (Theorem~\ref{thm:plus-randomness}) hold.
We thus get all of the same properties as the solution proven secure in that
theorem.  In terms of the fraction $\afrac$ of malicious participants that we
can tolerate, it is $t / n$, where $t$ is the threshold of the PVSS scheme
used to initialize the random beacon (see Appendix~\ref{appendix:scrape} for
more details).  After the random beacon is initialized, if implemented in
Ethereum our protocol then requires no further coordination, so can tolerate
any fraction.  In the context of proof-of-stake, however, we would still need
to assume an honest majority.
\end{proof}

Finally, in addition to the possibility that no participant ``wins'' a given
round, it is also possible that there will be two or more winners in a round.
In a pure leader election protocol, we could simply elect all winning
participants as leaders.  In a setting such as proof-of-stake, however, being
elected leader comes with a financial reward, and conflicts may arise if two
winners are elected (such as the nothing-at-stake problem discussed in
Section~\ref{sec:defns-pos}).  While we leave a full exploration of this as
future work (as again this issue does not arise in the abstract context of
leader election), one potential solution (also suggested by
Algorand~\cite{algorand}) for electing a single leader is to
require all participants to submit their
winning pre-image $h_{\round}$ before some time $t_{\round}$, and then select
as the winner the participant whose pre-image $h_{\round}$ has the lowest bit
value.

\section{Implementation and Performance}\label{sec:impl}
In this section, we present the implementation of \sysname as a suite of
Ethereum smart contracts, our experience running the protocol on a private
Ethereum test network,\footnote{Our full smart contract cannot be deployed on 
    the test network as blocks support only up to 4.7 million gas, whereas the
    production network supports 6.7 million gas.} 
and the technical difficulties we faced.
The main reason for doing so is to get an idea of the cost of such a protocol.

\subsection{Implementation details}\label{sec:impl-details}

We implemented the \sysname protocol as three smart contracts written in the 
Solidity language; together, they are approximately 950-1,000 lines of code.\footnote{ https://www.dropbox.com/s/n89xq4ta17a4n0p/Dealer.sol}\footnote{https://www.dropbox.com/s/yde7e5ja55wq1sb/Dealer\_voting.sol}
The first two contracts, \dealercon and \scrapecon, implement the SCRAPE 
random beacon of Cascudo and David~\cite{scrape} (which we describe in more
depth in Appendix~\ref{appendix:scrape}), and \caucuscon 
implements our full leader election protocol.
Furthermore, there is an optional \texttt{LocalCrypto}\footnote{https://www.dropbox.com/s/ajzgdw7v59ms43k/LocalCrypto.sol} contract adopted from \cite{mccorry2017smart} that contains the code for creating the zero knowledge proofs.

To set up \sysname, a single party creates the \caucuscon contract (Step~1) 
that establishes deadlines $t_{reg}, t_{com}, t_{scr}$  the maximum 
round time $t_{caucus}$, and the required registration deposit $d$.
Next, \caucuscon creates the \scrapecon contract (Step~2) before beginning 
the registration phase.  Participants run $\deposit$ (Step~3) to deposit $d$ 
coins, and also indicate whether or not they want to participate 
in SCRAPE.  Finally, registration is closed after $t_{reg}$ and \scrapecon 
creates a new \dealercon contract (Step~4) for each participant that 
indicated interest.  All dealers are set as participants in other dealer 
contracts (Step~5) before \scrapecon transitions to the commit phase and 
all \dealercon contracts transition to the distribution phase (Step~6). 

To now initialize the random beacon using SCRAPE, each dealer runs the secret
sharing phase of SCRAPE and publishes the resulting encrypted shares and
discrete log equality proofs (DLEQ) to the contract, which it then processes
(Step~7).  In our fully on-chain variant of \sysname, the \dealercon contract
also runs all necessary verifications before storing both the committed and 
encrypted shares (Step~8a).  

Because this type of on-chain verification is financially expensive (as it 
requires the contract to execute complex operations with high gas costs), we 
also consider a voting-based variant that adopts the techniques of
RandShare~\cite{randhound} to reduce on-chain costs but preserve liveness.  
In this variant, the contract stores the shares and proofs but doesn't perform 
any verification itself.  Instead, participants perform the verification 
themselves off-chain, and then cast a vote to indicate their support
accordingly (Step~8b).  The contract accepts encrypted and committed share 
once more than $n/2$ participants have voted that the proof is correct.

The \scrapecon contract transitions to the recovery phase after $t_{com}$, 
and is responsible for transitioning all dealer contracts that completed the 
distribution phase to the recovery phase (Step~9).  \caucuscon forfeits the 
deposit of any dealers that failed to finish the distribution phase.  In this
next phase, more than $n/2$ participants must publish their decrypted share 
and a DLEQ that it corresponds to the encrypted version in the respective 
\dealercon contract.  In our on-chain variant, the \dealercon contract 
verifies the DLEQ before storing the decrypted share (Step~10a).  In our
voting variant, the contract stores it regardless and then waits for an
appropriate number of positive votes about its validity (Step~10b).  Once a
sufficient number of shares (and votes, if applicable) are available, the
contract combines the shares to compute the dealer's secret (Step~11).  The 
\scrapecon contract then combines all dealer secrets to compute the final 
beacon $R_{\round}$ and transitions \caucuscon to the election phase (Step~12). 
\caucuscon forfeits the deposits of any participants that did not publish 
their decrypted share before $t_{src}$. 

At each round during the election phase, participants can run $\reveal$ to
send a value $h_\round$ to the contract.  The contract runs $\verify$ to check
that the participant has indeed won the election before computing the next 
beacon value $R_{\round+1} \gets R_{\round} \oplus h_{\round}$ and starting 
the next round (Step~13).  If there is no winner after $t_{caucus}$, then
the next value is computed as $R_{\round+1}\gets H(R_\round)$.

\subsection{Gas costs for cryptographic operations} \label{sec:overhead}


All cryptographic operations are implemented over the \verb#secp256k1# curve 
in Solidity.  The creation code for the DLEQ proofs needed for SCRAPE is 
distributed using the optional \texttt{LocalCrypto} contract and the DLEQ 
verification code is distributed in the \texttt{Dealer} contract. 
We briefly explore the number of multiplications and additions required for 
verifying the correct execution of the protocol (on-chain verification) before 
highlighting the gas that can be saved using the voting approach (off-chain 
verification). 

We focus on the gas costs for on-chain verification in both the distribution 
and recovery phases (Steps~8a and~10b respectively), in which participants send 
the contract DLEQ proofs. 
Each DLEQ proof involves 4 MUL and 2 ADD operations which is approximately 1.5-1.4M (million) gas.
Validating the $k$ committed shares published by the Dealer in Step~8a requires $k$ MUL and $k-1$ ADD operations which is approximately 2.1M gas when $k=6$ and the gas cost increases linearly for each additional committed share. 
Finally, combining all decrypted shares in Step~11 involves $k$ MUL and $k-1$ ADD operations which is approximately 550k gas when $k=7$. 

We focus on the gas costs for participants casting votes on whether the proofs are correct in both the distribution and recovery phases (Steps~7 and~9 respectively).
This requires $n^2-n$ transactions, since $n$ is the number of 
\texttt{Dealer} contracts and $n-1$ is the number of participants in each 
contract. 
A single vote to confirm that all encrypted and committed shares distributed
by the dealer are well formed requires approximately 42k gas.
On the other hand, a participant (and the respective dealer) must vote on
whether or not decrypted shares distributed by other participants are 
well formed.
A single vote on all 6 decrypted shares requires approximately 400k gas.
Although the voting approach involves an additional 76 transactions overall, 
it also requires 104M less gas.

\subsection{Financial costs} \label{sec:financial}

\begin{table*}[t]
\centering
\begin{tabular}{llS[table-format=9.0]S[table-format=3.2]S[table-format=3]S[table-format=8.0]S[table-format=2.2]S[table-format=3]}
\toprule
Step & Purpose & \multicolumn{3}{c}{On-chain protocol} & 
    \multicolumn{3}{c}{Voting protocol}\\
\cmidrule(lr){3-5} \cmidrule(lr){6-8}
&& {Cost (gas)} & {Cost (\$)} & {\# txns} & {Cost (gas)} & {Cost (\$)} 
    & {\# txns}  \\
\midrule
\textbf{Unit costs} & & & & & & & \\
\midrule
1 & Create \caucuscon & 6605787 & 12.3 & 1 & 5879425 & 10.9 &
    1 \\
2 & Create \scrapecon & 4271053 & 7.94 & 1 & 3562713 & 6.63 & 1 \\
3 & Register a participant & 167374 & 0.31 & 7 & 167374 & 0.31 & 7 \\
4 & Create \dealercon & 2652412 & 4.93 & 7 & 2652412 & 4.93 & 7 \\ 
5 & Set participants in all \dealercon  & 22613 & 0.04 & 1 & 29928 & 0.06 & 1 \\
6 & Transition \dealercon to distribution stage  & 87375 & 0.16 & 7 & 87245 & 0.16 & 7 \\
7 & Process encrypted shares in \dealercon & 1557849 & 2.90 & 42 & 255271 & 0.47 & 42 \\
8a & Validate committed shares for a dealer & 2166942 & 4.03 & 7 & {\dash}  & {\dash}  & {\dash}  \\
8b & Process votes on encrypted shares & {\dash} & {\dash} & {\dash} & 42671 & 0.08 & 42 \\
9 & Transition \scrapecon to recovery stage & 53269 & 0.10 & 1 & 61063 & 0.11 & 1 \\
10a & Process decrypted shares in \dealercon & 1413342 & 2.63 & 42 & 187640 & 0.35 & 42 \\
10b & Process votes on decrypted shares & {\dash} & {\dash} & {\dash} & 399127 & 0.74 & 42 \\
11 & Join all decrypted shares in a dealer's contract & 552034 & 1.03 & 7 & 552034 & 1.03 & 7 \\
12 & Transition \caucuscon to election stage & 49267 & 0.09 & 1  & 27202 & 0.05 &  1 \\
13 & Elect a leader (per round) & 53651 & 0.10 & 1 & 53651  &  0.10 & 1 \\
\midrule
\textbf{Aggregate costs} & & & & & & & \\
\midrule
& Create all contracts & 29443724 & 54.8 & 10 & 28009022 & 52.1 & 10 \\
& Initialize beacon & 145731246 & 271.0 & 115  & 42922542 &  79.8 & 192 \\
& Elect a leader (per round) & 53651 & 0.10 & 1  & 53651 & 0.10 & 1 \\
\midrule
\textbf{Total costs} && 175228621 & 325.9 & 126  & 70985215 & 132.0 & 202 \\
\bottomrule
\end{tabular}
\caption{\label{fig:financialcost} Breakdown of the gas costs for \sysname 
leader election with 7 participants, for both the on-chain and voting-based
variants.  The costs in USD (\$) were computed using a conversion rate of 1
ether = \$465 and 1 gas = 4 GWei.}
\end{table*}

Table \ref{fig:financialcost} presents the estimated financial costs for 
setting up and running \sysname on the Ethereum network, both in terms of the
gas costs associated with each individual step described in
Section~\ref{sec:impl-details}, as well as the
aggregate costs for each phase of the protocol.  The conversions used were
taken as their real-world values in December 2017 (\url{ethgasstation.info}).

The total cost for creating all contracts is \$54.8 in the on-chain variant
and \$52.1 in the voting-based variant.  This cost is incurred only once, as
the contract can be stored once in the blockchain and \texttt{DELEGATECALL} 
can re-use this code to create a new instance of each contract. 

The total cost for executing the SCRAPE protocol is \$271 in the on-chain
variant and \$79.8 in the voting-based variant.  After this initial round,
the total cost for running a single round of \sysname is then \$0.10, if 
the contract must perform only a single hash to verify $h_\round$ against 
the committed head of the hash chain; i.e., if $\round_\maxr - \round = 1$.  
The cost for each additional hash is 128 gas, which is roughly \$0.0001.

Overall, the initial cost to deploy \sysname's initial contracts is around \$50.  
If a secure random beacon is required from the very start, then the cost can be kept under \$80 using the voting-based approach.  
On the other hand, if a secure random beacon is not required immediately (i.e. the beacon is used to bootstrap a proof of stake protocol as we describe in Section~\ref{sec:bootstrap}), then it is possible to avoid SCRAPE altogether.
However, relying solely on Caucus assumes that at least one honest participant has contributed to the beacon before bootstrapping begins. 
Either way, the cost to continue the random beacon and thus elect new leaders is only \$0.10.

\section{Improvements and Extensions}\label{sec:improvements}

In this section we discuss how the basic \sysname protocol presented in
Section~\ref{sec:construction} can be changed to detect misbehaving 
participants, and how it can be integrated with existing blockchains based 
on proof-of-work (PoW) to enable proof-of-stake (PoS).

\subsection{Detecting bias}

In our current protocol specification, the winner of a leader election could
potentially bias the beacon by simply not revealing the value $h_{\round}$.
If the contract were updated to store a list of the recent values produced by
the random beacon and participants were required to periodically reveal a list
of pre-images, then the protocol could retroactively detect this behavior.  In
particular, it could check for each beacon value $R_k$ and for
all participants $i$ with value $h_k^{(i)}$ at the appropriate round whether 
or not $H(h_k^{(i)}\oplus R_k) < \target_k$, and also whether or not the
participant submitted $h_k^{(i)}$.  If the inequality holds but the value was
not submitted, then the contract can punish the participant appropriately by,
for example, taking their security deposit.




\subsection{Bootstrapping PoS from PoW} \label{sec:bootstrap}

The Ethereum Foundation plan to switch Ethereum's consensus protocol from
proof of work to proof of stake \cite{casper}, which motivates the need 
for a protocol that can facilitate this transition.  Our implementation 
demonstrates that \sysname can be set up as a smart contract to satisfy this 
transition at a reasonable cost, and because it facilitates an open-membership 
policy that allows anyone to register and participate in leader elections.  A 
proof-of-stake protocol could then take over as the new consensus protocol 
after numerous leader elections in \sysname, which would also increase the
possibility that a single honest participant had contributed to the random
beacon (and thus meaning we could avoid the expensive step of initializing it
using a coin-tossing protocol). 

The missing part of \sysname in order to present a full PoS protocol is a 
reward and punishment scheme that solves the nothing-at-stake and long-range 
attacks described in Section~\ref{sec:defns-pos}.  We have intentionally left 
this problem for future work to instead focus on rigorously analyzing our 
leader election protocol.

\section{Conclusion and Open Problems}
In this paper, we presented \sysname, a leader election protocol that, among
other potential applications, could be used to bootstrap a PoW-based blockchain 
to a PoS-based one.  To analyze \sysname, we presented cryptographic notions 
of security specific to the setting of blockchains; namely, unpredictability, 
fairness and liveness.  We then proved that \sysname satisfies strong variants
of these security properties, while requiring only
minimal coordination between the participants.
While other papers try to solve the proof-of-stake problem all at once, we 
instead focused on the security properties of the leader election, which is a 
well studied problem in the distributed system literature but one that needs a 
different formalization in the context of blockchains.  More generally, it is 
difficult to compare the different PoS protocols that have been proposed, and 
we believe our model makes steps towards providing a formal model that allows
for meaningful comparisons.

Due to our focus on leader election, we leave as a future work an economic 
incentivization scheme that would be needed to provide a full solution to the 
PoS problem, as proof-of-stake problem consists mostly of a secure leader 
election combined with an incentive-compatible reward and punishment scheme.

{\footnotesize
\def\shortbib{1}
\bibliographystyle{abbrv}
\bibliography{abbrev2,crypto,misc}
}

\appendix
\section{SCRAPE}\label{appendix:scrape}

The SCRAPE protocol~\cite{scrape} instantiates a publicly verifiable secret
sharing (PVSS) scheme in order to achieve a publicly verifiable random beacon.
As discussed in Section~\ref{sec:construction}, we have decided to incorporate
SCRAPE into \sysname due to its low computational complexity and its
compatibility with public ledgers.

Briefly, a $(n,t)$-threshold secret sharing scheme~\cite{shamir79} consists of
a set of $n$ participants $(P_1,\dots,P_n)$ and a dealer $D$ who distributes
shares $(s_1,\ldots,s_n)$ of a secret $s$ among the participants in such a way
that $t$ or more of them can recover the secret $s$.  A publicly verifiable 
secret sharing (PVSS) scheme~\cite{FOCS:CGMA85,C:Schoenmakers99} has the 
additional property that anyone can verify that the shares have been computed 
correctly and allow the participants to recover a valid secret.  This means
that PVSS schemes consist of a \emph{distribution} phase (in which encrypted 
shares are distributed and their correctness can be publicly verified) and a
\emph{recovery} phase (in which encrypted shares are decrypted and revealed,
and can be recovered in the event that some of the initial participants have
gone offline).

\subsection{The original SCRAPE protocol}

Briefly, the original SCRAPE protocol proceeds as follows: the parameters 
consist of
generators $g$ and $h$ of a group $\mathbb{G}$ of prime order $q$, as well as
a hash function $H(\cdot)$ modeled as a random oracle. $C$ is also defined to be 
the linear error correcting code corresponding to the $(n, t)$-threshold Shamir 
secret sharing scheme, and let $C^{\bot}$ be its dual code.

To register, each participant $P_{i}$ generates a secret key 
$\sk_i \randpick \F_q$ and a public key $\pk_{i} \gets h^{\sk_{i}}$, and
publishes $\pk_i$ to the ledger.

In the distribution phase, the dealer $D$ picks a secret $s$ and computes the
shares $(s_1,\ldots,s_n)$ of the secret, where $n$ is the number of registered
participants.  Then, each dealer $D_{i}$ publishes a list of encrypted shares 
$(\hat{s}_1,\dots,\hat{s}_n)$, where $\hat{s}_{i} = \pk_{i}^{s_{i}}$, a list of 
committed shares $(v_{1},\dots,v_{n})$, where $v_{i} = h^{s_{i}}$, and a list of 
discrete log equality proofs $DLEQ(g,v_{i},\pk_{i},\hat{s}_{i})$ that each pair of encrypted and committed shares correspond to the same discrete log.  
To verify these values, anyone can verify the list of DLEQ proofs before sampling random codewords $c^{\bot} = (c^{\bot}_{1},...,c^{\bot}_{n})$ 
of the dual code $C^{\bot}$, corresponding to the instance of Shamir's 
$(n,t)$-threshold secret sharing used by $D$, and verify that 
$\prod^{n}_{i=1} v_{i}^{c_{i}^{\bot}} = 1$.

In the recovery phase, each participant decrypts their share 
as $\tilde{s}_i = \hat{s}^{1/\sk_{i}}_{i}$ and publishes a new proof 
$DLEQ(h,pk_{i},h^{s_{i}},\hat{s}_{i})$ that the decrypted share corresponds 
to the encrypted share.  These proofs can again be verified by anyone.  
Finally, the dealer's secret $S$ can be recovered once half the shares are 
decrypted using Lagrange interpolation (where 
$\lambda_{i} = \prod_{i \neq j} \frac{j}{j-i}$ are the Lagrange coefficients);
i.e., by computing

\[
\prod_{P_{i} \in Q}^{n} (\tilde{s}_{i})^{\lambda_{i}} =	\prod_{P_{i} \in
Q}^{n} h^{p(i)} = h^{s}.
\]

More details and a proof of security (assuming DDH and modeling $H(\cdot)$ as
a random oracle) can be found in the original SCRAPE paper~\cite{scrape}.

\subsection{Our modified protocol}

To make the abstract SCRAPE protocol compatible with our Ethereum-based
implementation required several modifications that we highlight here.

First, as discussed in Section~\ref{sec:impl-details}, we included an 
additional timer $t_{com}$ that allows the contract to dictate the end of 
the distribution phase, whereas in the original SCRAPE protocol the recovery
phase is said to start after $n/2$ dealers have finished the distribution 
stage.  This was done in order to provide a grace period for all dealers to 
publish their commitment before transitioning to the recovery stage.  We still
require that $t = n/2$.

Second, the original SCRAPE protocol proposes as an optimization that the 
dealer could send an additional commitment to the full secret $S$ in the
distribution phase, and then open it in the recovery phase; this would mean
the other participants wouldn't have to recover it.  There is no proposed 
method in the paper, however, for the
dealer to prove that this additional commitment is the same secret that would
ultimately be derived from the decrypted shares.  Without such a proof, the 
dealer could commit to a second secret $S'$ and then choose to either allow 
the participants to recover $S$ or to reveal $S'$; this would in turn allow it
to bias the beacon.  To avoid this type of behavior, we chose not to include
this optimization.

Third, the original SCRAPE protocol suggests picking a random codeword in
order to perform verification.  In our fully on-chain variant, verification is
performed by the contract, which has no source of randomness on which to draw.
Instead then, it uses as a random seed the hash of the most recent block to do
this sampling.

Finally, all DLEQ proofs in SCRAPE rely on a single challenge $e$, which is 
the hash of all commitments, ciphertexts, and other random factors.  
It is not feasible to include all this information in a single transaction as a Solidity function can only support 16 local variables (including its parameters). 
Instead, the dealer sends one proof per transaction, along with
the value $e$ that the dealer has computed locally.
The contract can verify the proof using the dealer's supplied $e$ before storing the committed and encrypted share.
This is necessary to allow the contract to perform both the verification and storage in a single transaction. 
After verifying all proofs, the contract can re-compute $e'$ using the stored committed/encrypted shares before proceeding to the reveal phase if $e = e'$. 
Otherwise, it halts in a failed state. 



\end{document}